\documentclass[conference]{IEEEtran}

\usepackage{graphicx}
\usepackage{epstopdf}
\usepackage{amsmath, amssymb, amsthm, amsfonts}
\usepackage[font=footnotesize]{subfig}
\usepackage{mathtools}
\usepackage{tikz}
\usetikzlibrary{shapes}
\usepackage{pgfplots}
\tikzset{>=latex}
\usepackage{setspace}
\usepackage{afterpage}

\usepackage{cite}
\usepackage{multicol}

\usetikzlibrary{patterns}
\usetikzlibrary{automata, positioning}

\newtheorem{lemma}{Lemma}

\newcounter{tempEquationCounter} 
\newcounter{thisEquationNumber}
\newenvironment{floatEq}
{\setcounter{thisEquationNumber}{\value{equation}}\addtocounter{equation}{1}
	\begin{figure*}[!t]
		\scriptsize\setcounter{tempEquationCounter}{\value{equation}}
		\setcounter{equation}{\value{thisEquationNumber}}
	}
	{\setcounter{equation}{\value{tempEquationCounter}}
		\hrulefill\vspace*{4pt}
	\end{figure*}
	
}

\begin{document}
	
\bstctlcite{IEEEexample:BSTcontrol}

\title{Age of Information Performance of Multiaccess Strategies with Packet Management}

\author{
	\IEEEauthorblockN{Antzela~Kosta\IEEEauthorrefmark{1}, Nikolaos~Pappas\IEEEauthorrefmark{1}, Anthony~Ephremides\IEEEauthorrefmark{1}\IEEEauthorrefmark{2}, and~Vangelis~Angelakis\IEEEauthorrefmark{1}}
	\IEEEauthorblockA{\IEEEauthorrefmark{1} Department of Science and Technology, Link{\"o}ping University, Campus  Norrk{\"o}ping, 
		60 174, Sweden}
	\IEEEauthorblockA{\IEEEauthorrefmark{2} Electrical and Computer Engineering Department, University of Maryland, College Park, MD 20742\\
		E-mail: \{antzela.kosta, nikolaos.pappas, vangelis.angelakis\}@liu.se,  etony@umd.edu}  }

\maketitle

\begin{abstract}
We consider a system consisting of $N$ source nodes communicating with a common receiver.
Each source node has a buffer of infinite capacity to store incoming bursty traffic in the form of status updates transmitted in packets, which should maintain the status information at the receiver fresh. 
Packets waiting for transmission can be discarded to avoid wasting network resources for the transmission of stale information.
We investigate the age of information (AoI) performance of the system 
under scheduled and random access. 
Moreover, we present analysis of the AoI with and without packet management at the transmission queue of the source nodes, where as packet management we consider the capability to replace unserved packets at the queue whenever newer ones arrive. 
Finally, we provide simulation results that illustrate the impact of the network operating parameters on the age performance of the different access protocols. 
\end{abstract}

\IEEEpeerreviewmaketitle

\section{Introduction}
Future networks should support applications with heterogeneous QoS requirements, where critical performance indicators are the end-to-end delay, the throughput, the energy efficiency, and the service reliability.
The concept of \emph{age of information} (AoI) was introduced in \cite{Kaul11_SECON,Kaul11_GLOBECOM} to quantify the freshness of the knowledge we have about the status of a remote system. 
The \emph{age} captures the time elapsed since the last received message containing  update information was generated.
The novelty of this metric to characterize the freshness of information in a communication system differentiates it from other conventional metrics such as delay and connects it with emerging real-time wireless applications.

Maintaining data freshness is a requirement in numerous applications like wireless sensor networks (WSN) for healthcare and environmental monitoring, active data warehousing, energy harvesting \cite{Wu17_TGN,Arafa18_twc,Arafa18_arXiv_trans,Nath18_TGCom,Chen19_INFOCOM,Gu19_IoT}, web caching \cite{Yu99,Kam17_ISIT,Yates17_ISIT2,Zhong18_ISIT}, real time databases, ad hoc networks \cite{Kam2017_Milcom}, wireless smart camera networks \cite{He18_INFOCOM}, UAV-assisted IoT networks \cite{abd2018arXiv,Elmagid18_arXiv}, broadcast and multicast wireless networks \cite{Kadota2016_Allerton, Hsu17_ISIT,Kadota18_trans,Buyukates19_arXiv}, etc.
Moreover, in the field of adaptive transmission significant efficiency gains can be obtained by adaptive signaling strategies.
However, this feedback scheme is constrained by the acquisition of timely channel state information (CSI) \cite{Costa15_ICC, Costa15_ISIT, Klein17, Farazi17_ICCCN}.

The first attempts to address the AoI of a source at the destination of a status update transmission system were made through simple queueing models.
In \cite{Kaul12_INFOCOM}, three simple models were studied, the M/M/1, the M/D/1, and the D/M/1, under the first-come-first-served (FCFS) discipline. 
Alternative measures of stale information that are by-products of AoI are studied for the M/M/1 queue in \cite{Kosta17_ISIT}.
An expansion of the basic model that includes multiple sources sharing a common queue is considered in \cite{Modiano15_ISIT,Yates19_transactions,Stamatakis18_arXiv}.
The analysis therein illustrated how combining multiple sources in a common queue is more efficient in terms of the average AoI of each source, than serving them separately. 

Moving to different system characteristics, in \cite{Kam16} the authors consider different systems with either plentiful or limited network resources (servers).
Under this assumption, a more dynamic feature of networks is considered, that is, packets traveling over a network might reach the destination through numerous alternative  paths thus the delay of each packet might differ.
In this context, the performance of the M/M/1, M/M/2, and M/M/$\infty$ queues is provided, and the tradeoff between AoI and the waste of network resources in terms of non-informative packets as the number of servers varies, is demonstrated.

Two efficient ways to avoid congestion in networks are packet management techniques and admission control, since they can manage the traffic entering them.
Packet management by dropping or replacing packets was investigated in \cite{Costa16,Pappas15_ICC} where the M/M/1/1, M/M/1/2, and M/M/1/2* queues are considered.
A key outcome was that packet management can promote smaller average AoI, when compared to schemes without replacement and the same number
of servers.
The last-come-first-served (LCFS) queue discipline differs from packet management in that packets are not dropped if an infinite buffer is considered.
Allowing newly generated status updates to surpass older status updates, with and without the use of preemption, was studied in \cite{Kaul12_CISS,Najm16_ISIT,Bedewy16_ISIT,Yates18_INFOCOM,Najm18_INFOCOM,Yates18_arXivSHS,Inoue18_arXiv}.

Furthermore, a diversity of additional resource sharing features of a communication system have been studied in relation to AoI.
Transmission scheduling  is considered in \cite{He16_WiOpt,He16_ICC, Talak18_arXiv_SPAWC,Talak18_arXiv_Mobihoc,Talak18_arXiv_ISIT,Kadota18_transactions,Jiang2018_arXiv,Maatouk19_INFOCOM,Jiang19_INFOCOM} where centralized and decentralized scheduling policies for AoI minimization, under general interference constraints and time varying channels, are proposed.
The proposed scheduling algorithms have low complexity with strong AoI performances over stochastic information arrivals.
In \cite{Kaul17_ISIT} the authors consider scheduled access and slotted ALOHA-like random access, however the queueing aspect along with random access is not captured.
Throughput and AoI performance in a cognitive shared access network with queueing analysis has been studied in \cite{Kosta18_GLOBECOM}.
Additional references can be found in the survey \cite{NET-060}. 

\subsection{Contribution}
In this work, we focus on the AoI performance of a network consisting of $N$ source nodes communicating with a common receiver.
Each source node has a buffer of infinite capacity to store incoming bursty traffic in the form of packets which should keep the receiver timely updated.
We consider that the source nodes can discard packets waiting for transmission, in a process that is referred to as packet management.
We present analysis of the time average AoI with and without packet management at the transmission queue of the source nodes. 
We investigate three different policies to access the common medium (i) a round-robin scheduler (ii) a work-conserving scheduler 
(iii) random access.
To incorporate the effect of channel fading and network path diversity in such a system we provide simulation results that illustrate the impact of network operating parameters on the performance of the different access protocols.
The network path diversity refers to the transmission of packets over multiple alternate paths.

\section{System Model}
\label{sec:network_model}
We consider a wireless network consisting of $N$ source nodes communicating with a common receiver.
Each node has a buffer of infinite capacity to store incoming packets.
These packets are then sent through error-prone channels to the destination $d$, as shown in Fig.~\ref{fig:system_model}. 
Packets have equal length and time is divided into slots such that the transmission time of a packet from the buffer to the destination is equal to one slot.
Each such packet is said to provide a \emph{status update} and these two terms are used interchangeably.
The status updates arrivals are modeled by independent and identically distributed (i.i.d.) Bernoulli processes with average probabilities $\lambda_i \in (0,1)$, for
  $i=1,\dots,N$.
The probability distribution of time until successful delivery is assumed to be geometric with mean $1/\mu_i$ slots, for the $i$th node, where $\mu_i$ is referred as the service rate of the $i$th node.

We consider two different queue disciplines: \emph{without} and \emph{with} packet replacement. 
The first, assumes that all packets need to be delivered to the destination regardless of the freshness of the status update information.
We note that the motivation behind this discipline is in terms of the reconstructability of the transmitted information (that can also be related with estimation and prediction theory aspects) that is beyond the scope of this work.
The second discipline assumes that a packet which arrives while another packet is being served may be kept in the queue waiting for transmission. However, the packets waiting for transmission are replaced by newly generated packets of the same source.
We denote this discipline by \emph{replacement} and the process of discarding the packets from the queue is referred to as \emph{packet management}.
The packet management is expected to improve the performance of the system with respect to the staleness of the transmitted information.
Nevertheless, note that this is a non-conventional queueing model, for which some of the classic results from queueing theory, such as Little's law, do not apply \cite{Kleinrock}.

\begin{figure}[t!]
	\centering
	\includegraphics[scale=.52]{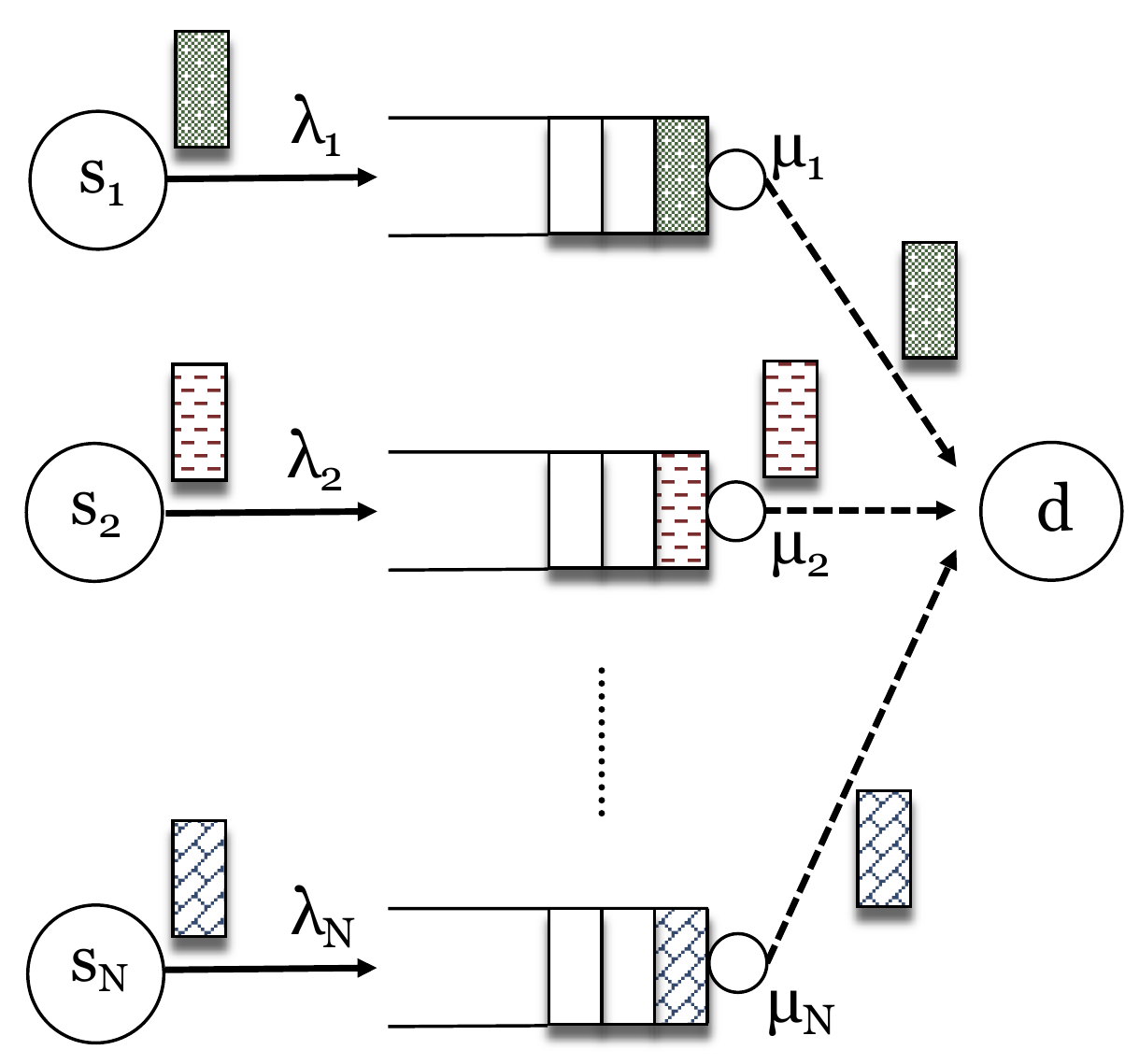}
	\caption{Status updates over a multiaccess network.}
	\label{fig:system_model}
\end{figure}

Status updates depart from the queues either in a perfectly scheduled or a random fashion.
We consider three different policies to access the common medium. 
\begin{itemize}
	\item \emph{Round-robin}: The scheduler assigns time slots to each node in equal portions and in fixed circular order.
	
	\item \emph{Work-conserving}: 
	The scheduler makes probabilistic decisions in each time slot, among the nodes that have a packet at the transmission queue.		

	\item \emph{Random}: The nodes attempt to transmit the packet at the head of the queue with a given probability $q_i$ colliding with each other.
\end{itemize}
These policies will be presented, evaluated, and compared in the next sections in terms of their AoI performance. 
 
\section{Age of Information Analysis}
\label{sec:AoI_model}

To derive the time average AoI of the system we start by characterizing AoI in terms of random variables that capture the age evolution at the receiver.
The age at the receiver depends on the packet receptions and the delay imposed by the network to these packets.
Then, expectations of the random variables are calculated for each of the queue disciplines separately.
In the next section, we evaluate AoI for the proposed access policies where the exact service rate at the queues is incorporated. 

Consider that the $j$th status update of node $i$ is generated at time $t_{ij}$, delivered through the transmission system, and received by the destination at time $t_{ij}^{'}$.
Then, we denote by $T_{ij} = t'_{ij}  - t_{ij}$ the system time of update $j$ of the $i$th node.
This corresponds to the sum of the queueing time and the queue service time.
The interarrival time of update $j$ of node $i$ is defined as the
random variable $Y_{ij} = t_{ij} - t_{i(j-1)}$.
Finally, let $Z_{ij} =  t'_{ij} - t'_{i(j-1)}$ be the random variable denoting the time between the reception of status update $(j-1)$ and $j$ of node $i$.

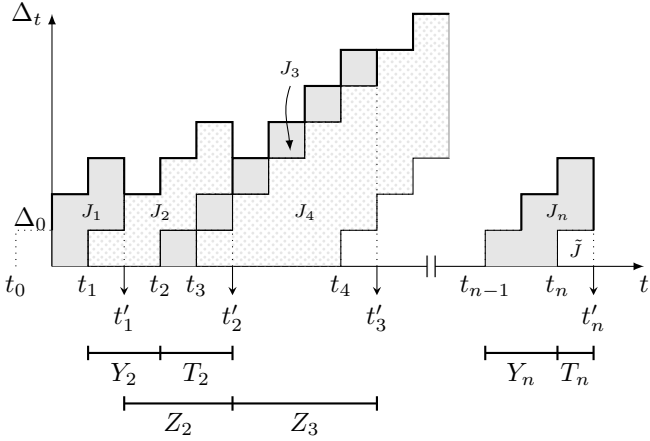
\begin{figure}[t!]
	\centering
	\begin{tikzpicture}[scale=0.96]
\draw[->] (0,0) -- (8.2,0) node[anchor=north] {$t$};
\draw[->] (0,0) -- (0,3.5) node[anchor=east] {$\Delta_t$};

\draw	(-0.3,0.38) node[anchor=south] {$\Delta_0$};

 

\draw[fill=gray!20] (0,0.0) -- (0,1.0) -- (0.5,1)-- (0.5,1.5)-- (0.5,1.5)-- (1.0,1.5)-- (1.0,0.5) -- (0.5,0.5) -- (0.5,0.0);
 
\draw[fill=gray!20]  (1.5,0) -- (1.5,0.5) -- (2.0,0.5) -- (2.0,1)-- (2.5,1) -- (2.5,1.5)-- (3.0,1.5)-- (3.0,2)-- (3.5,2.0)-- (3.5,2.5)-- (4.0,2.5)-- (4.0,3.0)-- (4.5,3.0)-- (4.5,2.5)-- (4.0,2.5)-- (4.0,2.0)-- (3.5,2)-- (3.5,1.5)-- (3.0,1.5)-- (3.0,1) -- (2.5,1)-- (2.5,0.5)-- (2,0.5)-- (2,0);

\draw[pattern=crosshatch dots, pattern color=gray!20] (0.5,0)-- (0.5,0.5)-- (1,0.5)--   (1.0,1)-- (1.5,1) -- (1.5,1.5)-- (2.0,1.5)-- (2.0,2)-- (2.5,2.0)-- (2.5,1.0)-- (2.0,1.0)-- (2.0,0.5)-- (1.5,0.5)-- (1.5,0.0);

\draw[pattern=crosshatch dots, pattern color=gray!20]  (2.0,0.0) -- (2.0,0.5) -- (2.5,0.5) -- (2.5,1.0) -- (3.0,1.0) -- (3.0,1.5) -- (3.5,1.5) -- (3.5,2.0) -- (4.0,2.0) -- (4.0,2.5) -- (4.5,2.5) -- (4.5,3.0) -- (5.0,3.0) -- (5.0,3.5)  -- (5.5,3.5) -- (5.5,1.5) -- (5.0,1.5) -- (5.0,1.0) -- (4.5,1.0) -- (4.5,0.5) -- (4.0,0.5)  -- (4.0,0.0);

\draw[white] (5.5,3.5) -- (5.5,1.5) ;

\draw[fill=gray!20] (6.0,0.0) -- (6.0,0.5) --  (6.5,0.5) -- (6.5,1)-- (7.0,1) -- (7.0,1.5)-- (7.5,1.5)-- (7.5,0.5) -- (7.0,0.5) -- (7.0,0.0);

\draw	(-0.5,0) node[anchor=north] {$t_0$}
           (0.5,0) node[anchor=north] {$t_1$}
		    (1.5,0) node[anchor=north] {$t_2$}
		    (2,0) node[anchor=north] {$t_3$}
		    (4.0,0) node[anchor=north] {$t_4$}
		    (6,0) node[anchor=north] {$t_{n-1}$}
		    (7.0,0) node[anchor=north] {$t_n$};
		    
\draw[->,>=stealth]    (1,0) -- (1,-0.4) node[anchor=south,below] {$t'_1$};
\draw[->,>=stealth]  (2.5,0) -- (2.5,-0.4) node[anchor=south,below] {$t'_2$};
\draw[->,>=stealth]  (4.5,0) -- (4.5,-0.4) node[anchor=south,below] {$t'_3$};
\draw[->,>=stealth]   (7.5,0) -- (7.5,-0.4) node[anchor=south,below] {$t'_n$};
		    	    		 
\draw	(0.55,0.75) node{{\scriptsize $J_1$}}
		    (1.5,0.75) node{{\scriptsize $J_2$}};
\draw   (3.5,0.75) node{{\scriptsize $J_4$}}
           (7,0.75) node{{\scriptsize $J_n$}};
 \draw   (7.25,0.25) node{{\scriptsize $\tilde{J}$}};
           
\draw[<-] (3.3,1.6) to [out=95,in=250] (3.3,2.5) node [above] {{\scriptsize $J_3$}};           
           
\draw [thick](0.5,-1.2) -- (1.5,-1.2) node[pos=.5,sloped,below] {$Y_2$} ;
\draw[thick]  (0.5,-1.3) -- (0.5,-1.1); 
\draw [thick](1.5,-1.2) -- (2.5,-1.2) node[pos=.5,sloped,below] {$T_2$} ;
\draw[thick]  (1.5,-1.3) -- (1.5,-1.1) 
                    (2.5,-1.3) -- (2.5,-1.1);
                    
\draw [thick](6,-1.2) -- (7.0,-1.2) node[pos=.5,sloped,below] {$Y_n$} ;
\draw[thick]  (6,-1.3) -- (6,-1.1); 
\draw [thick](7.0,-1.2) -- (7.5,-1.2) node[pos=.5,sloped,below] {$T_n$} ;
\draw[thick]  (7.0,-1.3) -- (7.0,-1.1) 
                    (7.5,-1.3) -- (7.5,-1.1);
                    
\draw [thick](1,-1.9) -- (2.5,-1.9) node[pos=.5,sloped,below] {$Z_2$} ;
\draw[thick]  (1,-2.0) -- (1,-1.8); 
\draw[thick]  (2.5,-2.0) -- (2.5,-1.8); 
\draw [thick](2.5,-1.9) -- (4.5,-1.9) node[pos=.5,sloped,below] {$Z_3$} ;
\draw[thick]  (4.5,-2.0) -- (4.5,-1.8);                    
 
\draw[thick] (0,0.5) -- (0,1.0) -- (0.5,1)-- (0.5,1.5)-- (0.5,1.5)-- (1.0,1.5)-- (1.0,1.0);

\draw[thick] (1.0,1)-- (1.5,1) -- (1.5,1.5)-- (2.0,1.5)-- (2.0,2)-- (2.5,2.0)-- (2.5,1.5);
\draw[thick] (2.5,1) -- (2.5,1.5)-- (3.0,1.5)-- (3.0,2)-- (3.5,2.0)-- (3.5,2.5)-- (4.0,2.5)-- (4.0,3.0)-- (4.5,3.0);
\draw[thick]  (4.5,2.5)-- (4.5,3.0)-- (5.0,3.0)-- (5.0,3.5)-- (5.5,3.5);

\draw[thick] (6.5,0.5) -- (6.5,1)-- (7.0,1) -- (7.0,1.5)-- (7.5,1.5)-- (7.5,0.5);

\draw[dotted] (1,0) -- (1,1.5);
\draw[dotted] (2.5,0) -- (2.5,2.0); 
\draw[dotted] (4.5,0) -- (4.5,3.0); 
\draw[dotted] (7.5,0) -- (7.5,0.5); 
\draw[dotted] (-0.5,0) -- (-0.5,0.5)-- (0,0.5); 
\draw[dotted] (0.5,0) -- (0.5,0.5) -- (1.0,0.5) -- (1.0,1); 
\draw[dotted] (1.5,0) -- (1.5,0.5) -- (2.0,0.5) -- (2.0,1)-- (2.5,1); 
\draw[dotted] (2,0) -- (2,0.5) -- (2.5,0.5) -- (2.5,1)-- (3.0,1) -- (3.0,1.5)-- (3.5,1.5)-- (3.5,2)-- (4.0,2.0)-- (4.0,2.5); 

\draw[dotted] (4,0) -- (4,0.5) -- (4.5,0.5)-- (4.5,1.0)-- (5.0,1.0); 
\draw[dotted] (6,0) -- (6,0.5) -- (6.5,0.5); 
\draw[dotted] (7.0,0) -- (7.0,0.5) -- (7.5,0.5); 

\draw[thick]  (5.2,-0.15) -- (5.2,0.15) 
                    (5.3,-0.15) -- (5.3,0.15);
\draw[white, fill=white!50] (5.21,-0.2) -- (5.21,0.2) -- (5.29,0.2) -- (5.29,-0.2) ;                   

\end{tikzpicture}
	\caption{Example of age evolution of node $i$ at the receiver.}
	\label{fig:age_vs_time_slotted}
\end{figure}

The AoI of each source node at the destination is defined as the random process $\Delta_t = t-u(t)$, where $u(t)$ is the timestamp of the most recently received update from that source.
An illustrative example of the evolution of the age of information of source $i$ in time is shown in Fig.~\ref{fig:age_vs_time_slotted}.
Without loss of generality, we assume that the observation of the system starts at $t=0$. 
At that time the queues are empty, and the AoI of the $i$th node at the destination is $\Delta_0$.
In the time intervals $[ t_{i(j-1)}^{'}, t_{ij}^{'}]$, $\forall j$,
the AoI increases in a stair-step fashion due to the absence of updates from node $i$ at the destination. 
Upon reception of a status update from node $i$ the AoI of that node is reset to a smaller value that is equal to the delay that the packet experienced.

Ensuring the average AoI of the $i$th node is small corresponds to maintaining information about the status of the node at the destination fresh.
For presentation clarity, from now on we drop the index denoting the source and focus on the packet index.
Given an age process $\Delta_t$ and assuming ergodicity, the average age can be calculated using a sample average that converges to its corresponding stochastic average.
For an interval of observation $(0,\mathcal{T})$, the time average age of node $i$ is
\begin{equation}
\Delta_{\mathcal{T}}=\frac{1}{\mathcal{T}}\sum_{t=0}^{N(\mathcal{T})}\Delta_t,
\label{eq:average_age}
\end{equation}
when we assume that the observation interval ends with the service completion of $N(\mathcal{T})$ samples.
The summation in \eqref{eq:average_age} can be calculated as the area under $\Delta_t$.
Then, the time average age can be rewritten as a sum of disjoint geometric parts.
Starting from $t=0$, the area is decomposed into the area $J_1$, the areas $J_j$ for $j=2, 3, \ldots N(\mathcal{T})$, and the area of width $T_n$ that we denote by $\tilde{J}$.
Then, the decomposition of $\Delta_{\mathcal{T}}$ yields
\begin{align}
\Delta_{\mathcal{T}} = & \frac{1}{\mathcal{T}} \left( J_1 + \tilde{J} + \sum_{j=2}^{N(\mathcal{T})} J_j \right) = \notag \\
= &  \frac{J_1 + \tilde{J}}{\mathcal{T}} + \frac{N(\mathcal{T})-1}{\mathcal{T}} \frac{1}{N(\mathcal{T})-1} \sum_{j=2}^{N(\mathcal{T})} J_j.
\label{eq:Delta_T}
\end{align}
The time average $\Delta_{\mathcal{T}}$ tends to the ensemble \textit{average age} as $\mathcal{T} \rightarrow\infty$, i.e.,
\begin{equation}
\Delta = \lim_{\mathcal{T} \rightarrow\infty}\Delta_{\mathcal{T}}.\footnote{We assume that the existence of the limit is guaranteed by the stability of the queues.}
\label{eq:Delta_avglimit}
\end{equation}
Note that the term $(J_1 + \tilde{J})/\mathcal{T}$ goes to zero as  $\mathcal{T}$ grows and also let
\begin{equation}
\lambda = \lim_{\mathcal{T} \rightarrow\infty} \frac{N(\mathcal{T})}{\mathcal{T}}
\label{eq:lambda}
\end{equation}
be the steady state rate of status updates generation.
Furthermore, using the definitions of the interarrival and system times, we can write the areas $J_j$ as
\begin{align}
J_j & = \sum_{m=1}^{Y_j+T_j} m - \sum_{m=1}^{T_j} m = \notag \\ 
& = \frac{1}{2} (Y_j+T_j)(Y_j+T_j+1) - \frac{1}{2} T_j(T_j+1) = \notag \\
& = Y_j T_j + Y_j^2/2 + Y_j/2.
\label{eq:J_i}
\end{align}
Then, substituting \eqref{eq:Delta_T}, \eqref{eq:lambda}, and \eqref{eq:J_i}, to \eqref{eq:Delta_avglimit} the average age of information of the $i$th node is given by
\begin{equation}
\Delta = \lambda\: \left( \mathbb{E}[YT]+\frac{\mathbb{E}[Y^2]}{2}+ \frac{\mathbb{E}[Y]}{2} \right),
\label{eq:av_prop}
\end{equation}
where $\mathbb{E}[\cdot]$ is the expectation operator.
The expression obtained in \eqref{eq:av_prop} differs from the expression obtained in \cite{Kaul12_INFOCOM} for the continuous time setup of the problem by an additional term $\mathbb{E}[Y]/2$.

Alternatively, we can express the areas $J_j$ with respect to the random variables $Z_j$, as follows  
\begin{align}
J_j & = \sum_{m=1}^{T_{j-1}+Z_j} m - \sum_{m=1}^{T_j} m  =\notag \\ 
& = \frac{1}{2} (T_{j-1}+Z_j)(T_{j-1}+Z_j+1) - \frac{1}{2} T_j(T_j+1),
\label{eq:J_i_v2temp}
\end{align}
and utilize the fact that when the system reaches steady state $T_{j-1}$ and $T_j$ are identically distributed.
We use $\mathbb{E}[T]$ to represent the expected value of $T_j$ for an arbitrary $j$.
Taking expectations of both sides gives
\begin{equation}
	\mathbb{E}[J] = \mathbb{E}[Z T] + \mathbb{E}[Z^2]/2 + \mathbb{E}[Z]/2.
	\label{eq:J_i_v2}
\end{equation}
Then, substituting \eqref{eq:Delta_T}, \eqref{eq:lambda}, and \eqref{eq:J_i_v2}, to \eqref{eq:Delta_avglimit} the average age of information of the $i$th node is given by
\begin{equation}
\Delta = \lambda\: \left( \mathbb{E}[ZT]+\frac{\mathbb{E}[Z^2]}{2}+ \frac{\mathbb{E}[Z]}{2} \right).
\label{eq:Delta_v2}
\end{equation}
In what follows, we analyze the steady-state age of information without and with packet management at the transmission queues.

	\begin{figure}[t!]
		\centering
		 \scalebox{.85}{
		 \begin{tikzpicture}
\node[state]             (s) {0};
\node[state, right=of s] (r) {1};
\node[state, right=of r] (r2) {2};
\node[state, right=of r2] (r3) {3};

\node[draw=none, right=of r3]   (r4)    {$\cdots$};

\draw[every loop, auto=left]
(s) edge[loop left] node {$(1-\lambda)$} (s)
(s) edge[bend left] node {$\lambda$} (r)
(r) edge[bend left] node {$\lambda(1-\mu)$} (r2)
(r2) edge[bend left] node {$\lambda(1-\mu)$} (r3)
(r3) edge[bend left] node {$\lambda(1-\mu)$} (r4);

\draw[every loop,auto=right]
(r4) edge[bend left, below] node {$\mu(1-\lambda)$} (r3)
(r3) edge[bend left, below] node {$\mu(1-\lambda)$} (r2)
(r2) edge[bend left, below] node {$\mu(1-\lambda)$} (r)
(r) edge[bend left, below] node {$\mu(1-\lambda)$} (s);           

\draw[] (r) edge[out=110, in=75
, looseness=0.8, loop
, distance=1.1cm, ->] node[above=0.5pt] {$1-r-s$} (r);

\draw[] (r2) edge[out=110, in=75
, looseness=0.8, loop
, distance=1.1cm, ->] node[above=0.5pt] {$1-r-s$} (r2);

\draw[] (r3) edge[out=110, in=75
, looseness=0.8, loop
, distance=1.1cm, ->] node[above=0.5pt] {$1-r-s$} (r3);

\end{tikzpicture}    }
		\caption{The DTMC which models the \emph{Geo/Geo/1} queue evolution at node $i$.}
		\label{fig:markov_chain_1}
\end{figure}
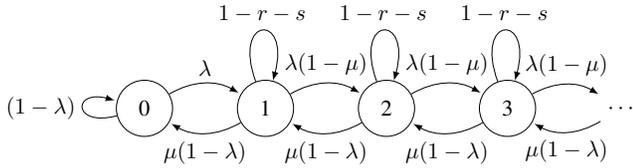

\subsection{Geo/Geo/1 Queue}
First, we derive the average AoI in \eqref{eq:av_prop} of the $i$th node without packet management, at the destination.
The interarrival times $Y_j$ are i.i.d. sequences that follow a geometric distribution therefore we know that 
\begin{align}
\mathbb{E}[Y_j] & =\frac{1}{\lambda},  \qquad
\mathbb{E}[Y_j^2]  = \frac{2-\lambda}{\lambda^2}.
\label{eq:E_Y_Y2}
\end{align}
Then, the only unknown term for the calculation of the average age is the expectation $\mathbb{E}[YT]$.
The system time of update $j$ is $T_j=W_j+S_j$, where $W_j$ and $S_j$ are the waiting time and service time of update $j$, respectively. 
Since the service times $S_j$ are independent of the interarrival times $Y_j$, we can write
\begin{equation}
\mathbb{E}[Y_j T_j] = \mathbb{E}[Y_j(W_j+S_j)] = \mathbb{E}[Y_j W_j] + \mathbb{E}[Y_j]\mathbb{E}[S_j],
\label{eq:expTiYi}
\end{equation}
where $\mathbb{E}[S_j] = 1/ \mu$.
Moreover, we can express the waiting time of update $j$ as the remaining system time of the previous update minus the elapsed time between the generation of updates $(j-1)$ and $j$, i.e.,
\begin{equation}
W_j = (T_{j-1} - Y_j)^+.
\label{eq:W_i}
\end{equation}
Note that if the queue is empty then $W_j = 0$.
Also note that when the system reaches steady state the system times are stochastically identical, i.e., $T =^{st} T_{j-1} =^{st} T_j$.

In addition, the queue of the $i$th node can be described through a discrete-time Markov chain (DTMC), where each state represents the number of packets in the queue. 

\begin{lemma}\label{lemma1}
	From the DTMC described in Fig.~\ref{fig:markov_chain_1} we obtain the following steady state probabilities
    \begin{equation}
    \pi_n = \rho^{n-1} \pi_1, \quad n \geq 1   \quad \text{and} \quad \pi_0 = \frac{\mu(1-\lambda)}{\lambda} \pi_1,
    \label{eq:pi_i_geo}
    \end{equation}
	where $\rho = \frac{\lambda(1-\mu)}{\mu(1-\lambda)}$, $\pi_1 = \frac{\lambda(1-\rho)}{\mu}$, $r=\lambda(1-\mu)$, and $s=\mu(1-\lambda)$.
\end{lemma}

To derive the probability mass function (pmf) of the system time $T$, we use the fact that the sum of $N$ geometric random variables where $N$ is geometrically distributed is also geometrically distributed, according to the convolution property of their generating functions \cite{Nelson2013probability}.
Let $S_j$, $j = 1,2, .. $ be independent and identically distributed geometric random variables with parameter $\mu$.
If an arriving packet sees $N$ packets in the system, then, the system time of that packet, using the memoryless property, can be written as the random sum $T = S_1 + \dots + S_{N}$.
To calculate the probability generating function of $T$ we condition on $N=n$ which occurs with probability $(1-\rho)\rho^{n-1}$ and obtain 
\begin{align}
G_T(z) &= \sum_{n=1}^{\infty} \left( \frac{\mu z}{1-(1-\mu)z}\right) ^n (1-\rho) \rho^{n-1} = \notag \\
&=   \frac{\mu (1-\rho) z}{1-(1-\mu (1-\rho))z}.
\label{eq:G_T}
\end{align}
This implies that the system time pmf is given by 
\begin{equation}
f_T(t)= \mu (1-\rho) (1-\mu+\mu \rho)^{t-1}.
\label{eq:f_T}
\end{equation}
Hence, $T$ follows a geometric distribution with parameter $\mu (1-\rho)$. 
An alternative approach that uses moment generating functions can also be found in \cite{Talak18_arXiv_Mobihoc}.

Now we are able to compute the conditional expectation of the waiting time $W_j$ given $Y_j=y$ as
\begin{align}
\mathbb{E}[&W_j|Y_j=y] = \mathbb{E}[(T_{j-1}-y)^+|Y_j=y] = \mathbb{E}[(T-y)^+] = \notag \\
& = \sum_{t=y}^\infty (t-y) f_T(t) =
\frac{(1-\mu+\mu\rho)^y}{\mu (1-\rho)}.
\label{eq:Wi_cond_Yi}
\end{align}
Then, the expectation $\mathbb{E}[W_j Y_j]$ is obtained as
\begin{align}
\mathbb{E}[W_j Y_j] &=  \sum_{y=0}^\infty y\: \mathbb{E}[W_j | Y_j=y] \: f_{Y_j}(y) = \notag \\ & = \frac{ \lambda (1-\mu+\mu \rho)}{\mu (1-\rho) (\lambda+\mu-\lambda \mu-\mu \rho+\lambda \mu \rho)^2}.
\label{eq:expWiYi}
\end{align}
Substituting $\rho = \frac{\lambda(1-\mu)}{\mu(1-\lambda)}$ to  \eqref{eq:expWiYi} and after some algebra we obtain
\begin{equation}
\mathbb{E}[W_j Y_j] = \frac{ \lambda (1-\mu)}{(\mu- \lambda) \mu^2}.
\label{eq:expWiYi2}
\end{equation}
From \eqref{eq:expWiYi2}, \eqref{eq:expTiYi}, and \eqref{eq:av_prop}, the average AoI of the $i$th node is obtained as
\begin{equation}
\Delta_{\text{Geo/Geo/1}}  = \frac{1}{\lambda}+\frac{1-\lambda}{ \mu-\lambda}-\frac{\lambda}{\mu^2}+\frac{\lambda}{\mu}.
\label{eq:Delta}
\end{equation}

In order to find the optimal value of $\lambda$ that minimizes the average AoI we proceed as follows.
We differentiate \eqref{eq:Delta} with respect to $\lambda$ to obtain $\frac{\partial \Delta}{\partial \lambda}$. 
By setting $\frac{\partial \Delta}{\partial \lambda}=0$ we can obtain the value of $\lambda$ that minimizes the AoI and satisfies the equation $\lambda^4 (-1+\mu) - 2 \lambda^3 (-1+\mu) \mu- \lambda^2 \mu^2 +2 \lambda \mu^3 -\mu^4 = 0$.
Trivially one can see that $\Delta$ is a convex function of $\lambda$ for a given service rate $\mu$, if $\lambda<\mu$ is not violated, by taking the second derivative $\frac{\partial^2 \Delta}{\partial \lambda^2}$.

\subsection{Queue with Replacement}
Next, the queue with replacement at the $i$th node can be described as a three-state discrete-time Markov chain where each state represents an empty system, a single packet receiving service, or a packet in the queue waiting for a packet in the server, respectively, as in \cite{Costa16}. 
The packet replacement does not affect the number of packets in the system since a newly generated packet discards the packet waiting in the queue, if any.

\begin{lemma}\label{lemma2}
       From the DTMC described in Fig.~\ref{fig:markov_chain_2} we obtain the following steady state probabilities
       \begin{equation}
       \pi_n = \frac{\lambda^n (1-\mu)^{n-1}}{\mu^n (1-\lambda)^n} \pi_0, \quad  n \in \{1,2\},  
       \label{eq:pi_n}
       \end{equation}
       \begin{equation}
       \text{and} \quad \pi_0 = \frac{\lambda-\mu}{\lambda \rho^2 -\mu},
       \label{eq:pi_0}
       \end{equation}
       where $\rho = \frac{\lambda(1-\mu)}{\mu(1-\lambda)}$, $r=\lambda(1-\mu)$, and $s=\mu(1-\lambda)$.
\end{lemma}
\begin{proof} See Appendix~\ref{Appendix_A'}. \end{proof}

	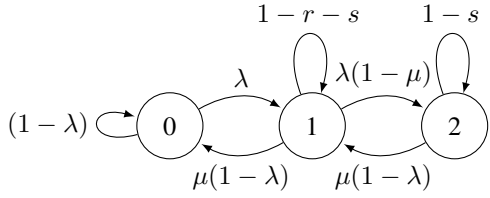
\begin{figure}[t!]
		\centering
		 \begin{tikzpicture}
\node[state]             (s) {0};
\node[state, right=of s] (r) {1};
\node[state, right=of r] (r2) {2};

\draw[every loop, auto=left]
(s) edge[loop left] node {$(1-\lambda)$} (s)
(s) edge[bend left] node {$\lambda$} (r)
(r) edge[bend left] node {$\lambda(1-\mu)$} (r2);

\draw[every loop,auto=right]
(r2) edge[bend left, below] node {$\mu(1-\lambda)$} (r)
(r) edge[bend left, below] node {$\mu(1-\lambda)$} (s);           

\draw[] (r) edge[out=110, in=75
, looseness=0.8, loop
, distance=1.1cm, ->] node[above=0.5pt] {$1-r-s$} (r);

\draw[] (r2) edge[out=110, in=75
, looseness=0.8, loop
, distance=1.1cm, ->] node[above=0.5pt] {$1-s$} (r2);

\end{tikzpicture}    
		\caption{The DTMC which models the evolution of the queue \emph{with replacement} at node $i$.}
		\label{fig:markov_chain_2}
\end{figure}

To calculate the average AoI of node $i$ at the destination for the replacement queue discipline we use \eqref{eq:Delta_v2} and describe an event such that $Z_j$ and $T_{j-1}$ are conditionally independent. 
In general, the inter-reception time $Z_j$ depends on the system time $T_{j-1}$ of the previous packet in the system and this complicates the analysis of their joint distribution. 
We denote by $\psi_j$ the event that the system is empty after the $j$th successful transmission.
Furthermore, let $\bar{\psi}_{j}$ be the complementary event that the $j$th packet leaves behind a system with a packet waiting in the queue.
The normalized probabilities of these events are given by
\begin{equation}
\mathbb{P} (\psi_{j}) = \frac{\pi_0}{\pi_0+\pi_1} = \frac{\mu-\lambda\mu}{\lambda+\mu-\lambda\mu},
\label{eq:P_psi}
\end{equation}
\begin{equation}
\mathbb{P} (\bar{\psi}_{j}) = \frac{\pi_1}{\pi_0+\pi_1} = \frac{\lambda}{\lambda+\mu-\lambda\mu}.
\label{eq:P_psi_bar}
\end{equation}
Then, the expectations $\mathbb{E}[ZT]$,  $\mathbb{E}[Z]$, and $\mathbb{E}[Z^2]$, in \eqref{eq:Delta_v2} can be calculated by conditioning on the events $\psi_j$ and $\bar{\psi}_{j}$. 

The inter-reception time of the $j$th packet given that the $(j-1)$th packet leaves behind an empty system is given by the convolution of two independent geometric random variables that represent the interarrival time of update $j$ and the service time of the same update. 
Hence,  
\begin{align}
&\mathbb{P} \{Z_j=z|\psi_{j-1}\} = \sum_{k=1}^{z-1} \mathbb{P} \{Y_j=k\} \mathbb{P} \{S_j=z-k\} = \notag \\
&=\frac{\lambda \mu}{\mu - \lambda} \left[  (1-\lambda)^{z-1} - (1-\mu)^{z-1} \right],
\label{eq:pmf_Zi1}
\end{align}
\begin{equation}
\mathbb{E}[Z_j | \psi_{j-1}] = \frac{\lambda+\mu}{\lambda\mu},
\label{eq:exp_Zi1}
\end{equation}
\begin{equation}
\mathbb{E}[Z_j^2 | \psi_{j-1}] = \frac{2\lambda^2 + 2 \lambda\mu -\lambda^2 \mu +2\mu^2 - \lambda\mu^2}{\lambda^2 \mu^2}.
\label{eq:exp_Zi12}
\end{equation}
Moreover, in case there is a packet waiting in the queue that starts service as soon as packet $(j-1)$ completes service, we have 
\begin{equation}
\mathbb{P} \{Z_j=z|\bar{\psi}_{j-1}\} = \mu \:(1-\mu)^{z-1},
\label{eq:pmf_Zi2}
\end{equation}
\begin{equation}
\mathbb{E}[Z_j | \bar{\psi}_{j-1}] = \frac{1}{\mu},
\label{eq:exp_Zi2}
\end{equation}
\begin{equation}
\mathbb{E}[Z_j^2 | \bar{\psi}_{j-1}] = \frac{2-\mu}{\mu^2}.
\label{eq:exp_Zi22givenpsibar}
\end{equation}
Then, the last two terms in \eqref{eq:Delta_v2} can be obtained as
\begin{align} 
\mathbb{E}[Z_j] & = \mathbb{E}[Z_j | \psi_{j-1}] \mathbb{P}(\psi_{j-1}) + \mathbb{E}[Z_j | \bar{\psi}_{j-1}] \mathbb{P}(\bar{\psi}_{j-1})  = \notag \\
&= \frac{\lambda+\mu}{\lambda\mu} \frac{\mu-\lambda\mu}{(\lambda+\mu-\lambda\mu)} 
+ \frac{1}{\mu} \frac{\lambda}{(\lambda+\mu-\lambda\mu)} = \notag \\
&= \frac{\lambda^2(1-\mu)+\lambda(1-\mu)\mu+\mu^2}{\lambda \mu (\lambda+\mu-\lambda\mu)},
\label{eq:exp_Zi}
\end{align}
and
\begin{align}
\mathbb{E}[Z_j^2]  &= \mathbb{E}[Z_j^2 | \psi_{j-1}] \mathbb{P}(\psi_{j-1}) + \mathbb{E}[Z_j^2 | \bar{\psi}_{j-1}] \mathbb{P}(\bar{\psi}_{j-1})  = \notag \\
&=  \frac{2\lambda^2 + 2 \lambda\mu -\lambda^2 \mu +2\mu^2 - \lambda\mu^2}{\lambda^2 \mu^2} \frac{\mu-\lambda\mu}{(\lambda+\mu-\lambda\mu)} +  \notag \\ 
&+ \frac{2-\mu}{\mu^2} \frac{\lambda}{(\lambda+\mu-\lambda\mu)}.
\label{eq:exp_Zi22}
\end{align}

To derive the conditional distributions of service time given the events $\psi_{j-1}$ and $\bar{\psi}_{j-1}$ we note that the $(j-1)$th packet leaves behind an empty system if and only if zero arrivals occur while it is being served.
Then, the conditional distribution of service time given the event $\psi_{j-1}$, where $f_S(\cdot)$ is the service time pmf, is given by

\begin{align}
&\mathbb{P} \{S_{j-1}=k|\psi_{j-1}\}  = \frac{\mathbb{P}(\psi_{j-1} | S_{j-1}=k) f_S(k)}{\sum_{k=1}^{\infty} \mathbb{P}(\psi_{j-1} | S_{j-1}=k) f_S(k)}= \notag \\
&= \frac{ \binom{k}{0} (1-\lambda)^{k} \mu (1-\mu)^{k- 1}}{\sum_{k=1}^{\infty} \binom{k}{0} (1-\lambda)^{k} \mu (1-\mu)^{k- 1}} = \notag \\
&= ((1-\lambda)(1-\mu))^{k-1} (\lambda+\mu-\lambda\mu),
\label{eq:f_S_given_psi}
\end{align}
with the resulting conditional expectation

\begin{equation}
\mathbb{E}[S_{j-1} | \psi_{j-1}] = \frac{1}{\lambda+\mu-\lambda\mu}.
\label{eq:exp_SgivenPsi}
\end{equation}
For the complementary event $\bar{\psi}_{j-1}$ the conditional distribution of the service time is given by 
\begin{align}
&\mathbb{P} \{S_{j-1}=k|\bar{\psi}_{j-1}\}  = \frac{\mathbb{P}(\bar{\psi}_{j-1} | S_{j-1}=k) f_S(k)}{\sum_{k=1}^{\infty} \mathbb{P}(\bar{\psi}_{j-1} | S_{j-1}=k) f_S(k)}= \notag \\
&= \frac{ (1 - \binom{k}{0} (1-\lambda)^{k}) \mu (1-\mu)^{k-1} }{\sum_{k=1}^{\infty} (1 - \binom{k}{0} (1-\lambda)^{k}) \mu (1-\mu)^{k-1} } = \notag \\
&=\frac{ (1 - (1-\lambda)^{k}) \mu (1-\mu)^{k-1} (\lambda+\mu-\lambda\mu)}{\lambda},
\label{eq:f_S_given_psi_bar}
\end{align}
with the resulting conditional expectation

\begin{align}
\mathbb{E}[S_{j-1} | \bar{\psi}_{j-1}] &= \frac{ \lambda (1-\mu)^2 +(2-\mu)\mu}{ \mu (\lambda+\mu-\lambda\mu)} = \notag \\ 
&= \frac{1}{\lambda+\mu-\lambda\mu} + \frac{1}{\mu} -1.
\label{eq:exp_SgivenPsi_bar}
\end{align}

We proceed with the characterization of the waiting time for transmitted packets via considering the events of transmission (tx) or replacement (drop). 
We consider two possible server states of node $i$, either idle or busy.
A packet arrival finds the server idle with probability $\mathbb{P}(\text{idle})=\pi_0$, due to the BASTA property (Bernoulli Arrivals See Time Averages) \cite{Cooper}.
A packet arrival finds the server busy with probability $\mathbb{P}(\text{busy})= 1-\pi_0$.
This packet will receive service if and only if zero arrivals occur while the packet in the server is transmitted.
Let $R$ represent the remaining service time of an update, with pmf $f_R(r)$, and let $\phi$ be the event that zero arrivals occur during the remaining service time.
For every measurable set $A \subset [0,\infty)$, we define the probability
\begin{equation}
\mathbb{P}(\phi , R \in A) = \sum_{r \in A} \mathbb{P}(\phi | R=r) f_R(r).
\label{eq:pr_phi_R_MM12star}
\end{equation}
Then, the probability of transmission conditioned on the event that the server is busy is given by
\begin{align}
\mathbb{P}(\text{tx}| \text{busy})& = \sum_{r=0}^{\infty} \mathbb{P}(\phi | R=r) f_R(r) =\notag \\
&=  \sum_{r=1}^{\infty} \binom{r}{0} (1-\lambda)^{r} \mu (1-\mu)^{r- 1} = \notag \\
&= \frac{\mu - \lambda \mu}{ \lambda+\mu-\lambda\mu}.
\label{eq:pr_tx_given_busy_MM12star}
\end{align}
As a result,
\begin{equation}
\mathbb{P}(\text{busy}, \text{tx}) =  (1-\pi_0) \frac{\mu-\lambda \mu}{ \lambda+\mu-\lambda\mu}.
\label{eq:pr_busy_tx}
\end{equation}

The distribution of the waiting time conditioned on the event $\{ \text{busy, tx} \}$ is given by
\begin{align}
&f(w | \text{busy},\text{tx}) = f(r | \phi) = \frac{\mathbb{P}(\phi | R=r) f_R(r)}{\sum_{r=0}^{\infty} \mathbb{P}(\phi | R=r) f_R(r)}= \notag \\
&= \frac{ \binom{r}{0} (1-\lambda)^{r} \mu (1-\mu)^{r- 1}}{\sum_{r=1}^{\infty} \binom{r}{0} (1-\lambda)^{r} \mu (1-\mu)^{r- 1}} = \notag \\
&=  \left[ (1-\lambda)(1-\mu)\right]^{r-1} (\lambda+\mu-\lambda\mu) = \notag \\
&=  \left[ 1- (\lambda+\mu-\lambda\mu)\right]^{r-1} (\lambda+\mu-\lambda\mu).
\label{eq:f_w_given_busy_tx}
\end{align}
Hence, $W$ conditioned on the event $\{ \text{busy, tx} \}$ is geometrically distributed with parameter $(\lambda+\mu-\lambda\mu)$.

Finally, using \eqref{eq:f_w_given_busy_tx} the expected value of the waiting time for a transmitted packet is obtained as

\begin{align}
	\mathbb{E}[W | \text{tx}] &= \frac{ (1-\lambda) \lambda (1-\mu) (\mu+\lambda-2\lambda\mu)}{ \lambda^2(\mu-1)^2+\lambda(1-2\mu)\mu+\mu^2} \times \notag \\ 
	& \times \frac{1}{(\lambda+\mu-\lambda\mu) }.
	\label{eq:exp_w_given_busy_tx}
\end{align}

Next, given the conditional expectations of the service time \eqref{eq:exp_SgivenPsi} and \eqref{eq:exp_SgivenPsi_bar}, and the expectation of the waiting time \eqref{eq:exp_w_given_busy_tx}, we calculate the conditional expectations of the system time as follows
\begin{align}
\mathbb{E}[T _{j-1} | \psi _{j-1}] &= \mathbb{E}[W _{j-1} | \psi _{j-1}] + \mathbb{E}[S _{j-1} | \psi _{j-1}]= \notag \\
&=\mathbb{E}[W _{j-1}] + \mathbb{E}[S_{j-1} | \psi _{j-1}]= \notag \\
&= \frac{1+ \frac{(1-\lambda)\lambda(1-\mu)(\lambda+\mu-2\lambda\mu)}{\lambda^2(\mu-1)^2+\lambda(1-2\mu)\mu+\mu^2}}{(\lambda +\mu-\lambda \mu)},
\label{eq:exp_T_given_psi}
\end{align}
\begin{align}
&\mathbb{E}[T _{j-1} | \bar{\psi}_{j-1}] = \mathbb{E}[W _{j-1} | \bar{\psi}_{j-1}] + \mathbb{E}[S _{j-1} | \bar{\psi}_{j-1}]= \notag \\
&=\mathbb{E}[W _{j-1}] + \mathbb{E}[S_{j-1} | \bar{\psi}_{j-1}]= \notag \\
&= \frac{ (1-\lambda) \lambda (1-\mu) (\mu+\lambda-2\lambda\mu)}{(\lambda+\mu-\lambda\mu) (\lambda^2(\mu-1)^2+\lambda(1-2\mu)\mu+\mu^2)} +  \notag \\ &\quad+\frac{1}{\lambda+\mu-\lambda\mu} +\frac{1}{\mu} -1.
\label{eq:exp_T_given_psibar}
\end{align}

Utilizing the probabilities \eqref{eq:P_psi}, \eqref{eq:P_psi_bar}, the conditional expectations of the system time \eqref{eq:exp_T_given_psi}, \eqref{eq:exp_T_given_psibar}, and the conditional expectations of the inter-reception time \eqref{eq:exp_Zi1}, \eqref{eq:exp_Zi2}, we calculate $\mathbb{E}[T_{j-1} Z_j]$ as follows
\begin{align}
\mathbb{E}[T_{j-1} Z_j]&
= \mathbb{P}(\psi_{j-1}) (\mathbb{E}[Z_j |\psi_{j-1}] \mathbb{E}[T_{j-1} |\psi_{j-1}]) \notag \\
&\quad \quad + \mathbb{P}(\bar{\psi}_{j-1}) (\mathbb{E}[Z_j |\bar{\psi}_{j-1}] \mathbb{E}[T_{j-1}|\bar{\psi}_{j-1}]) =\notag \\  
&=\frac{1}{\mu^2}+\frac{1-\lambda}{\lambda\mu}-\frac{1+\lambda}{ (\lambda+\mu-\lambda\mu)^2}+ \frac{1+2\lambda}{\lambda+\mu-\lambda\mu}+ \notag \\ 
&+\frac{\lambda(1-2\mu+\lambda(3\mu-2))}{\lambda^2 (\mu-1)^2+ \lambda(1-2\mu) \mu +\mu^2}.
\label{eq:Timinus1Zi}
\end{align}

We refer to the time average rate of packets that enter and remain in the system as the \emph{effective rate} and define it as 
\begin{align}
\lambda_e &= \lambda (1-p_D) = \notag \\
&= \lambda - \lambda \frac{\lambda^2 (1-\mu)}{\lambda^2 (1-\mu) +\lambda(1-\mu)\mu+\mu^2},
\label{eq:effective_rate}
\end{align}
where $p_D$ is the packet dropping probability

\begin{equation}
p_D = \frac{\lambda^n (1-\mu)}{\mu^n (1-\lambda)}  
\left(1+\frac{\lambda}{\mu (1-\lambda)} +\frac{\lambda^n (1-\mu)^{n-1}}{\mu^n (1-\lambda)^{n-1}} \right)^{-1}
\end{equation}
for $n=2$.


Finally, using \eqref{eq:exp_Zi}, \eqref{eq:exp_Zi22}, \eqref{eq:Timinus1Zi}, \eqref{eq:effective_rate}, and \eqref{eq:Delta_v2}, the average age of information of node $i$ for the replacement discipline is calculated as shown in \eqref{eq:av_Delta_replace}.
We recall that the analysis provided herein does not consider any coupling between the transmission queues but instead focuses on the AoI performance of an independent queue.
Such a step would require knowing the stationary probability distribution of the joint queue length process.
We proceed in the next section by detailing the three proposed access policies and evaluating them through simulations.

\begin{floatEq}
	\begin{align} 
	\Delta_{\text{replacement}} &= \frac{1}{\lambda^2 (1-\mu)+\lambda (1-\mu)\mu+\mu^2} \Bigg( \lambda \mu (\lambda+\mu - \lambda \mu) \bigg(\frac{\lambda^2 (1-\mu)+\lambda (1-\mu) \mu+\mu^2}{2 \lambda \mu (\lambda+\mu -\lambda \mu)}+\frac{\lambda (\lambda (3 \mu-2)-2 \mu+1)}{\lambda^2 (\mu-1)^2+\lambda \mu (1-2 \mu)+\mu^2} \notag \\ 
	&+\frac{\lambda^3 (\mu-2) (\mu-1)+\lambda^2 (\mu-2) (\mu-1) \mu+\lambda \mu^2 (2-3 \mu)+2 \mu^3}{2 \lambda^2 \mu^2 (\lambda +\mu-\lambda\mu)}+\frac{1-\lambda}{\lambda \mu}+\frac{2 \lambda+1}{\lambda+\mu - \lambda \mu}-\frac{\lambda+1}{(\lambda+\mu - \lambda \mu)^2}+\frac{1}{\mu^2} \bigg) \Bigg).
	\label{eq:av_Delta_replace} 
	\end{align}
\end{floatEq}

\section{Simulation results}
\label{sec:scheduled access}
The objective considered in this paper is to minimize the time average AoI over all policies and all nodes.
In that direction, we first investigate all policies without the effect of channel fading through simulations.
We develop a MATLAB-based behavioural simulator where each case runs for $10^6$ timeslots.

\subsection{Round-robin Scheduled Access}
In the round-robin scheduler nodes take turns to transmit their status updates.
If there is no packet at the $i$th queue waiting for transmission then the assigned time slot to source $i$ is wasted with no transmission taking place.
Round-robin is a simple scheduler that does not require dynamic coordination but comes with a throughput loss.
Assuming a fixed scheduling interval, each node is assigned a unique time slot index. 

In Fig.~\ref{fig:age_vs_lambda_round_robin_classic} the average AoI per source is shown as a function of the arrival rate per source without any packet management, for $\lambda_1=\dots=\lambda_N$ and success probability 1, at the destination. 
We observe that the AoI tends to infinity as the arrival rate tends to $1/N$.
This is due to the violation of the stability conditions for the queues implying infinite queueing delay.

\begin{figure}[t!]\centering
	\centering
	\includegraphics[draft=false,scale=.5]{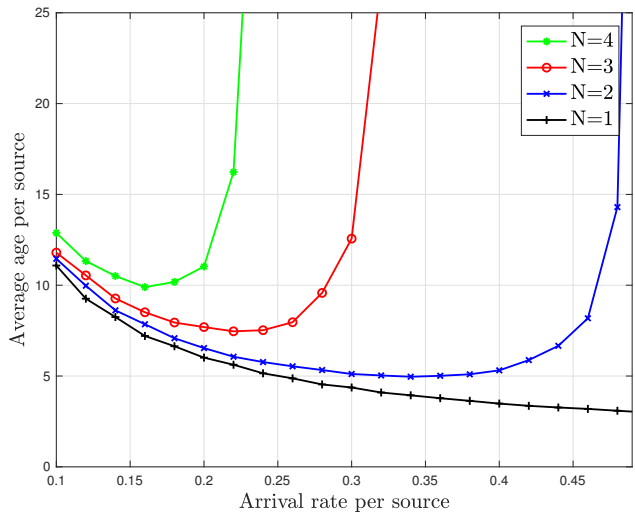}
	\caption{Average age per source vs. the arrival rate $\lambda_i$ for the round-robin scheduler without packet management at the transmission queues.}
	\label{fig:age_vs_lambda_round_robin_classic}
\end{figure}

In Fig.~\ref{fig:age_vs_lambda_round_robin_replace} the average AoI per source is shown as a function of the number of source nodes $N$ with a replacement queue, for $\lambda_1=\dots=\lambda_N$ and success probability 1, at the destination. 
In this case, the AoI is a monotonically decreasing function of the arrival rate.
Moreover, we note that the AoI increases linearly with the number of source nodes $N$.
	The average AoI for the round-robin scheduler with the replacement queue discipline is lower bounded by $\frac{N+3}{2}$, where $N$ the number of source nodes in the system, i.e.,  
	\begin{equation}
	\Delta_i \geq \frac{N+3}{2}, \quad \forall i \in \{1,\cdots,N\},  \quad \lambda_i \in  ( 0,1 ).
	\label{eq:Delta_min_RR}
	\end{equation}

\begin{figure}[t!]\centering
	\centering
	\includegraphics[draft=false,scale=.5]{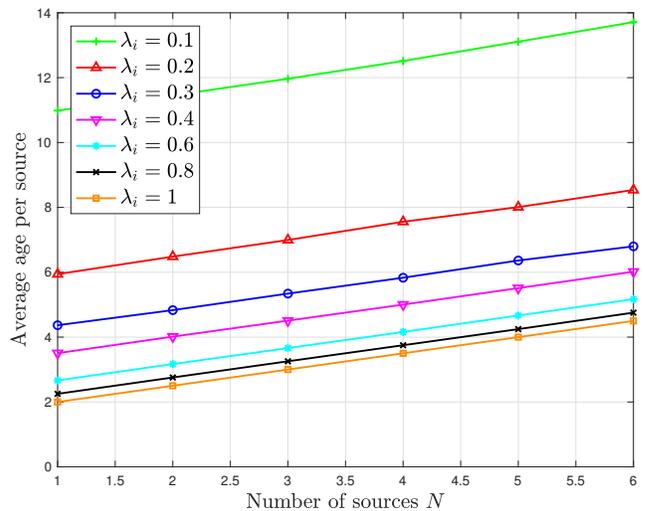}
	\caption{Average age per source vs. the number of sources $N$ for the round-robin scheduler with the replacement queue discipline.}
	\label{fig:age_vs_lambda_round_robin_replace}
\end{figure}

\subsection{Work-conserving Scheduled Access}

The work-conserving scheduler makes probabilistic decisions among the nodes that have a packet at the transmission queue at the same time slot.	
Specifically, source $i$ is assigned the given time slot with probability $1/\tilde{N}$, where $\tilde{N}$ is the number of sources that have a packet available for transmission.
A time slot is wasted with no transmission taking place only when we have an empty system.

In Fig.~\ref{fig:age_vs_lambda_work_conserving_replace} the average AoI per source is shown as a function of the number of source nodes $N$ with a replacement queue discipline, for $\lambda_1=\dots=\lambda_N$ and success probability 1, at the destination. 
With solid line we plot the work-conserving scheduler and with dashed line the round-robin scheduler.
The AoI of source $i$ for the work-conserving scheduler is a monotonically decreasing function of the arrival rate $\lambda_i$, similar to the round-robin scheduler.
Moreover, we observe that as $\lambda_i$ decreases, the gap between the performance of the work-conserving scheduler and the round-robin scheduler increases.
For $\lambda_i=1$ when there is always a packet available for transmission the performance of the two schedulers with respect to the AoI metric coincides.
	Therefore, the average AoI for the work-conserving scheduler with the replacement queue discipline is also lower bounded by $\frac{N+3}{2}$ where $N$ the number of source nodes in the system, i.e.,  
	\begin{equation}
	\Delta_i \geq \frac{N+3}{2}, \quad \forall i \in \{1,\cdots,N\},  \quad \lambda_i \in  (0,1).
	\label{eq:Delta_min_WR}
	\end{equation}

\begin{figure}[t!]\centering
	\centering
	\includegraphics[draft=false,scale=.5]{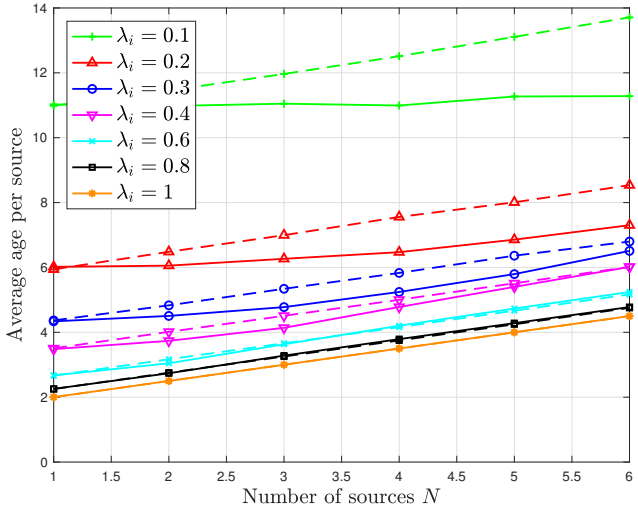}
	\caption{Average age per source vs. the number of sources $N$ for the work-conserving scheduler (solid lines) with the replacement queue discipline. The round-robin scheduler is depicted with dashed lines.}
	\label{fig:age_vs_lambda_work_conserving_replace}
\end{figure}

\subsection{Random Access}
\label{sec:random access}
In the slotted random access policy, at each time slot, node $i$ attempts to transmit the packet at the head of the queue with probability $q_i$, provided that the queue is not empty. 

\begin{figure}[t!]\centering
	\centering
	\includegraphics[draft=false,scale=.5]{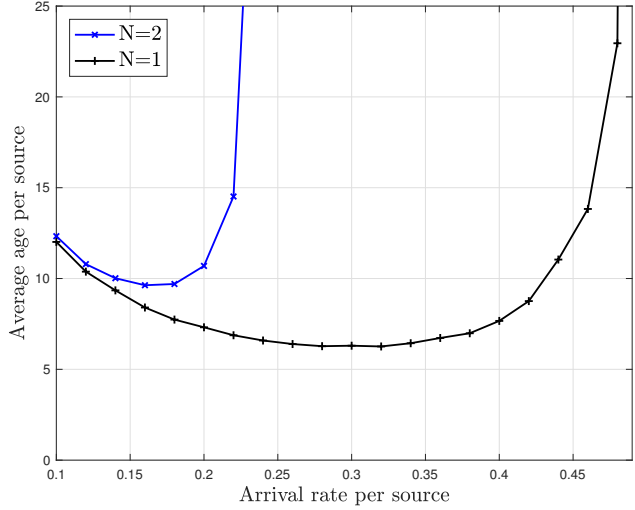}
	\caption{Average age per source vs. the arrival rate $\lambda_i$ for the random access without packet management at the transmission queues.}
	\label{fig:age_vs_lambda_slotted_ALOHA_classic}
\end{figure}

\begin{figure}[t!]\centering
	\centering
	\includegraphics[draft=false,scale=.5]{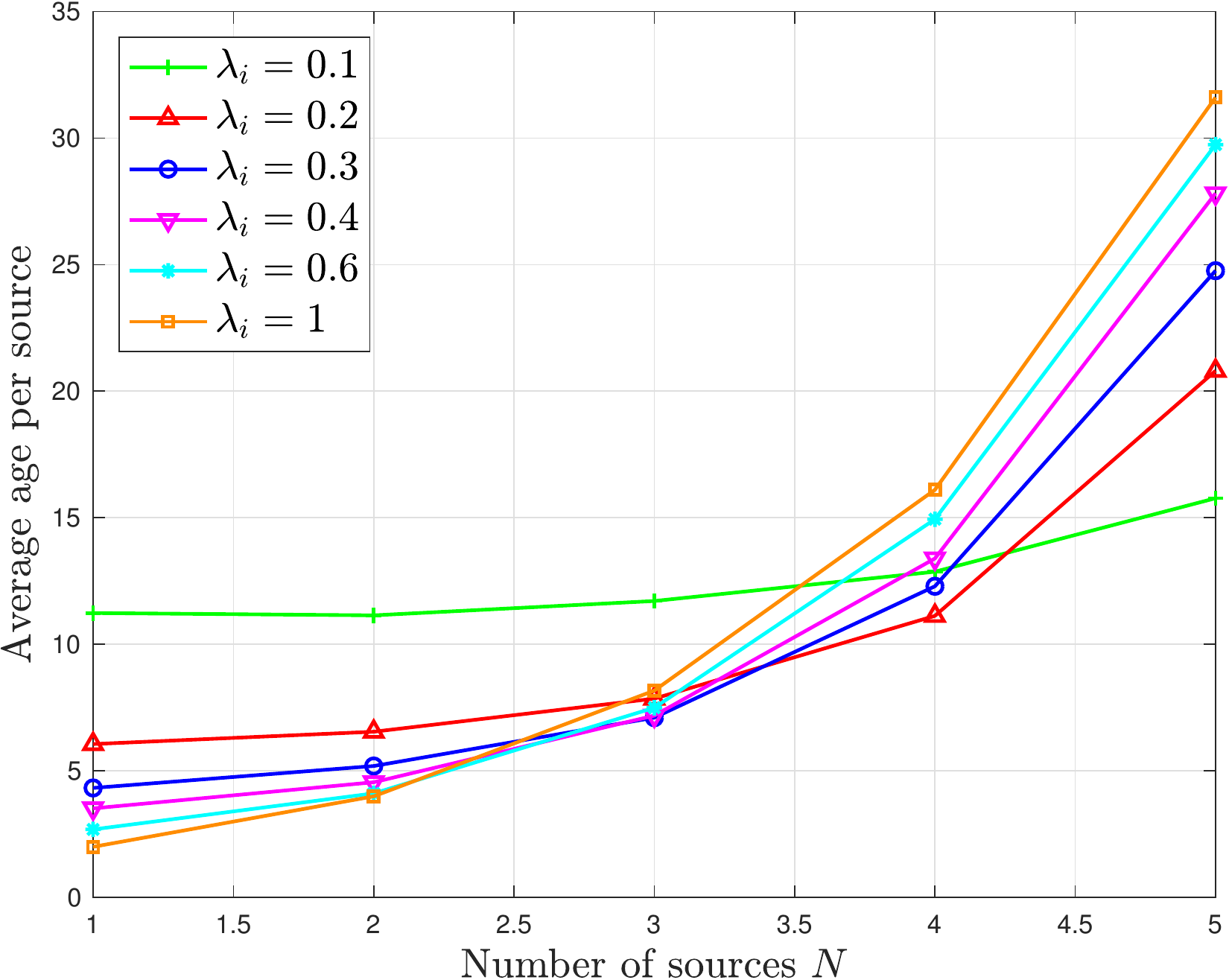}
	\caption{Average age per source vs. the arrival rate $\lambda_i$ for the random access with the replacement queue discipline.}
	\label{fig:age_vs_lambda_slotted_ALOHA_replace}
\end{figure}

\begin{figure}[t!]\centering
	\centering
	\includegraphics[draft=false,scale=.5]{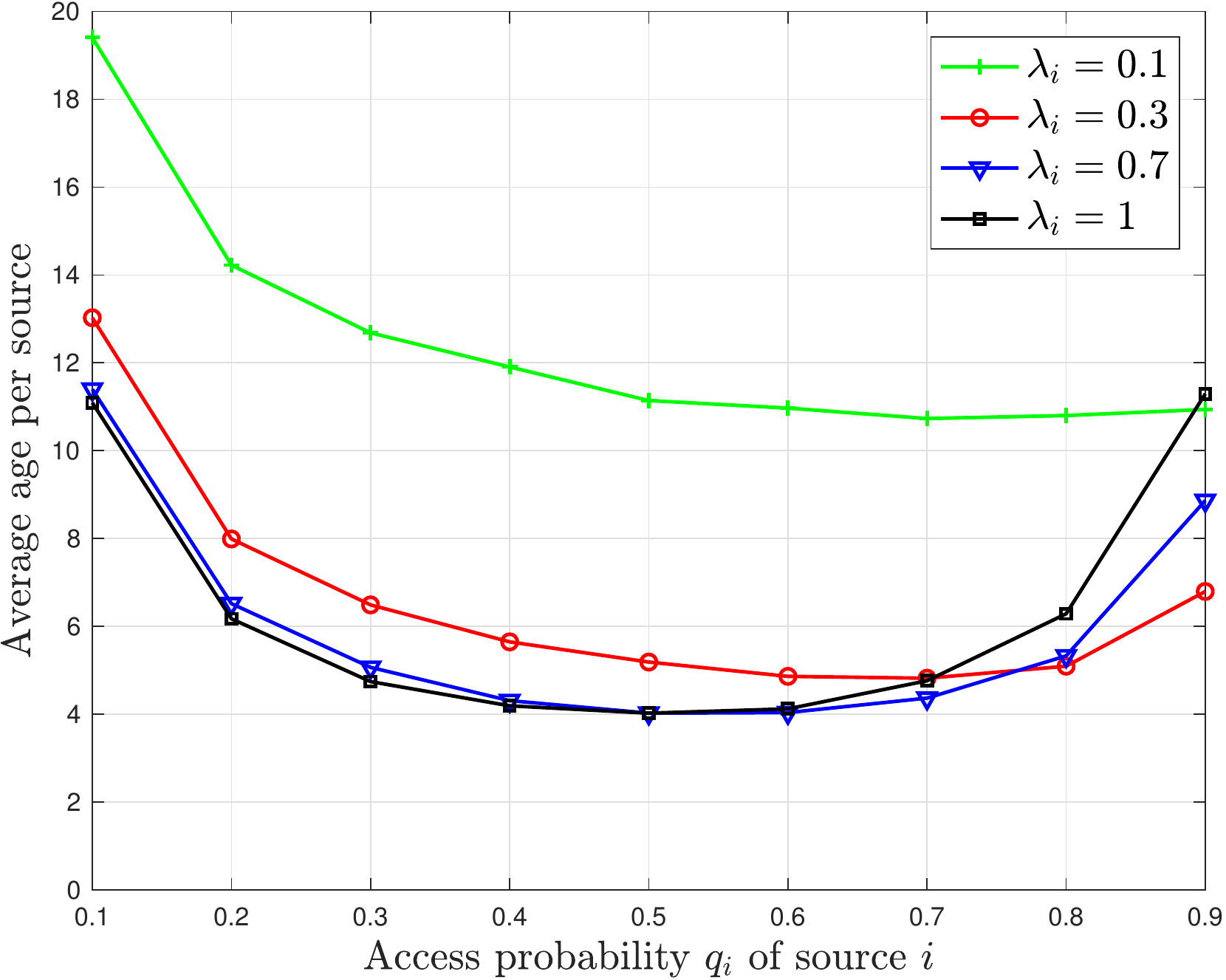}
	\caption{Average age vs. the access probability $q_i$ for the random access with the replacement queue discipline and $N=2$.}
	\label{fig:age_vs_p_slotted_ALOHA_replace}
\end{figure}

In Fig.~\ref{fig:age_vs_lambda_slotted_ALOHA_classic} the average AoI per source is shown as a function of the arrival rate per source without any packet management, for $\lambda_1=\dots=\lambda_N$, $q=q_1=\dots=q_N=0.5$, and a collision channel with success probability $\tilde{N} q (1-q)^{\tilde{N}-1}$, at the destination. 
We observe that the AoI tends to infinity as the arrival rate tends to $(1/N)*q$.
This is due to the violation of the stability conditions for the queues
implying infinite queueing delay.

In Fig.~\ref{fig:age_vs_lambda_slotted_ALOHA_replace} the average AoI per source is shown as a function of the number of source nodes $N$ with a replacement queue, for $\lambda_1=\dots=\lambda_N$, $q=q_1=\dots=q_N=0.5$, and a collision channel with success probability $\tilde{N} q (1-q)^{\tilde{N}-1}$, at the destination. 
We observe that the arrival rate that minimizes the AoI changes depending on the number of sources $N$.
Specifically, for small values of the arrival rate $\lambda_i$ it is preferable to have more source nodes transmitting, while for large values of the arrival rate $\lambda_i$ it is preferable to have few source nodes.
The AoI of source $i$ is a monotonically decreasing function of the arrival rate $\lambda_i$ for $N \in \{1,2\}$.

In Fig.~\ref{fig:age_vs_p_slotted_ALOHA_replace} the average AoI per source is shown as a function of the access probability $q_i$ of source $i$ with a replacement queue, for $N=2$, $\lambda_1=\lambda_2$, and a collision channel, at the destination. 
In this setup, we can find the optimal access probability $q_i$ for various arrival rates $\lambda_i$ and number of nodes $N$.
It is interesting to see that when $q_i$ is small it is better to have a large arrival rate in order to guarantee that there will be packets available for transmission.
On the other hand, for large $q_i$ a small rate is beneficial since the absence of packets reduces the collisions. 

\section{Fading and network path diversity}
\label{sec:comparison}

In this section, we consider the effect of the success probability of a packet erasure model and the effect of the network path diversity on the system, and present how the different parameters affect the system performance.

In particular, we assume that the node considered until now as the destination is an access point (AP).
Packets are sent through wireless channels to the AP and then from the AP they are transmitted through an error free network to the final destination, as shown in Fig.~\ref{fig:system_model2}. 
After the AP we consider a process that captures the \emph{network delay} imposed to packets. 
This is a simplified model of the random delay experienced by a packet after its departure from the AP to the final destination $d$. 
This process can model several cases, such as the delay for contending with other packets in the reception queue, multiple hops, or the processing time at the receiver. 
We model the availability of resources and the network path diversity by assuming an infinite number of servers at the AP.
The network delay process follows a geometric distribution with mean $1/k$, for $0<k<1$, and it causes packets to arrive at the destination $d$ out of order.

\begin{figure}[t!]
	\centering
	\includegraphics[scale=.41]{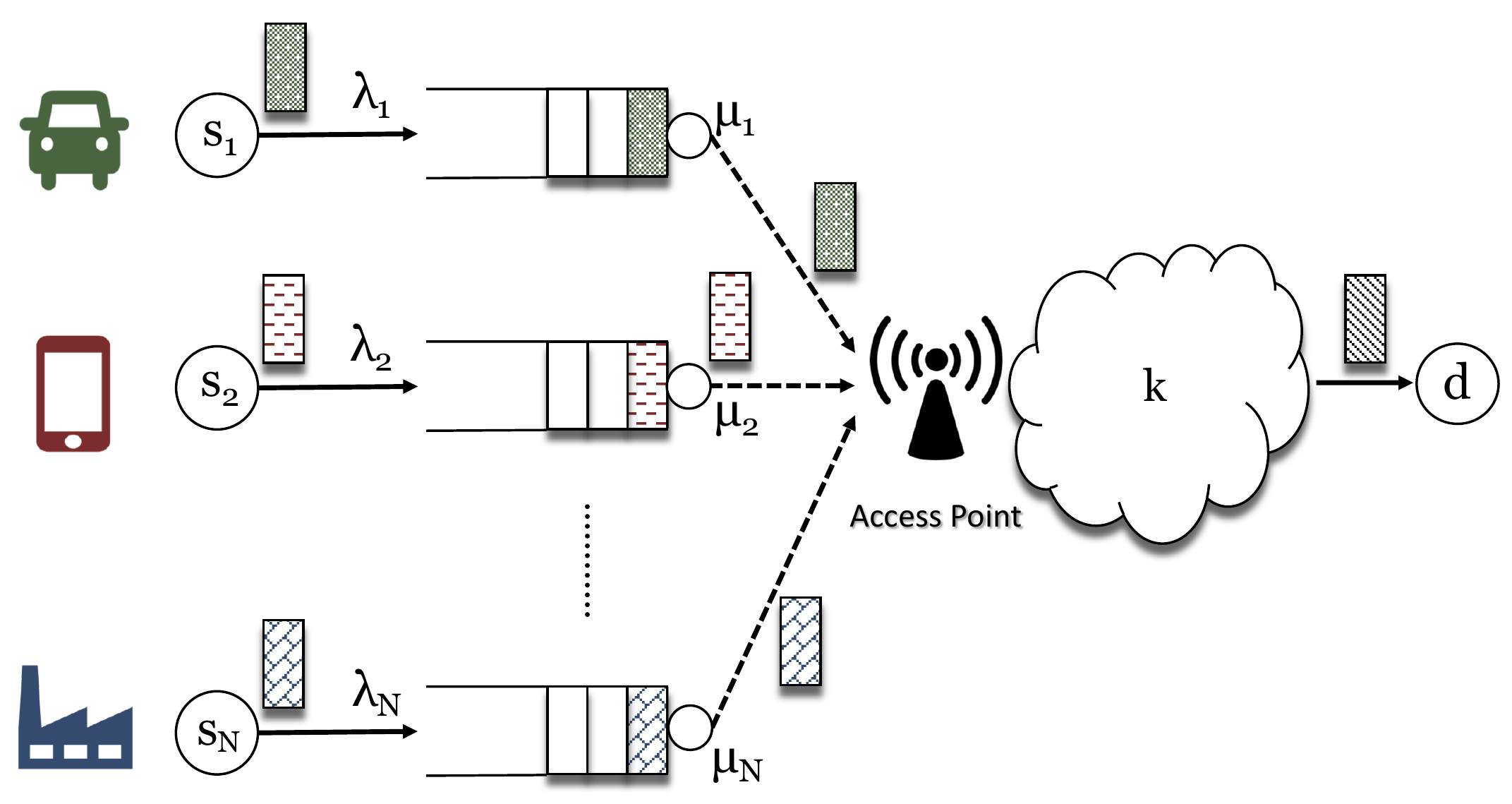}
	\caption{Status updates over a multiaccess network with out- of-order receptions.}
	\label{fig:system_model2}
\end{figure}

The network delay process can cause out of order reception of packets at the destination $d$.
We define an informative packet as a packet that carries the newest information compared to the packets of the same source arriving at the destination prior to it.
A packet $j$ is said to be obsolete if there is at least one packet with $k \geq 1$ of the same source generated after $j$, such that $t'_j > t'_{j+k}$. 
An informative packet is one that is not rendered obsolete.
Obsolete packets correspond to waste of resources since they do not provide fresh information to the destination. 
Thus, it is meaningful to minimize the percentage of obsolete packets among the transmitted packets.

In Fig.~\ref{fig:age_vs_lambda_round_robin_replace_withFading} the average AoI per source is shown as a function of the arrival rate per source node with a replacement queue, for the round-robin scheduler, $\lambda_1= \cdots =\lambda_N$, and success probability $p_1= \cdots =p_N$, at the AP. 
We see that as $p_i$ decreases, the gap between the performance of the system for $N=2$ and $N=3$ increases.
In other words, under good channel conditions adding more source nodes will degrade the AoI performance less compared to the case where the channel conditions are weak.

\begin{figure}[t!]\centering
	\centering
	\includegraphics[draft=false,scale=.5]{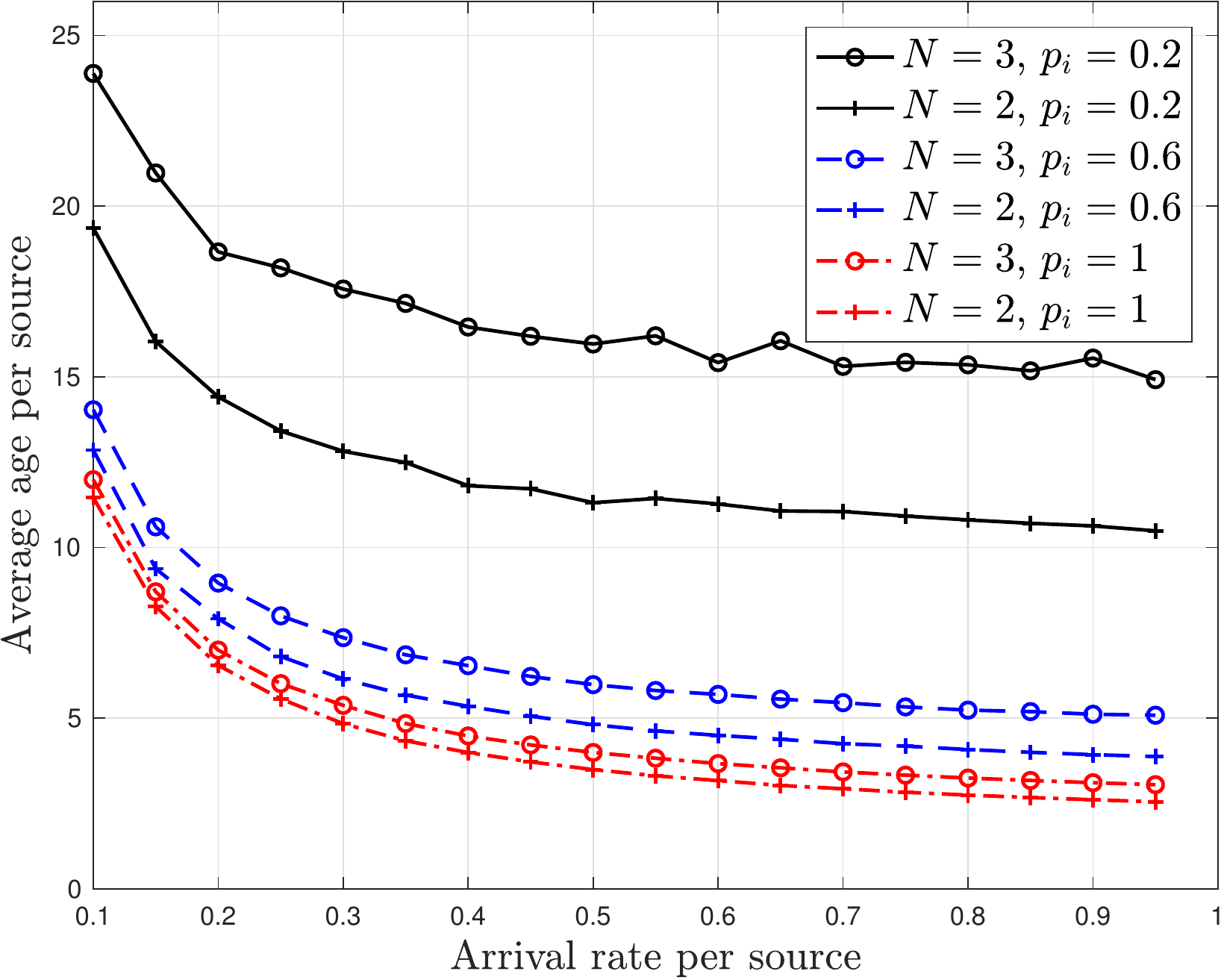}
	\caption{Average age per source vs. the arrival rate $\lambda_i$ for the round-robin scheduler with the replacement queue discipline. The success probability of the $i$th node is $p_i$.}
	\label{fig:age_vs_lambda_round_robin_replace_withFading}
\end{figure}

\begin{figure}[t!]\centering
	\centering
	\includegraphics[draft=false,scale=.5]{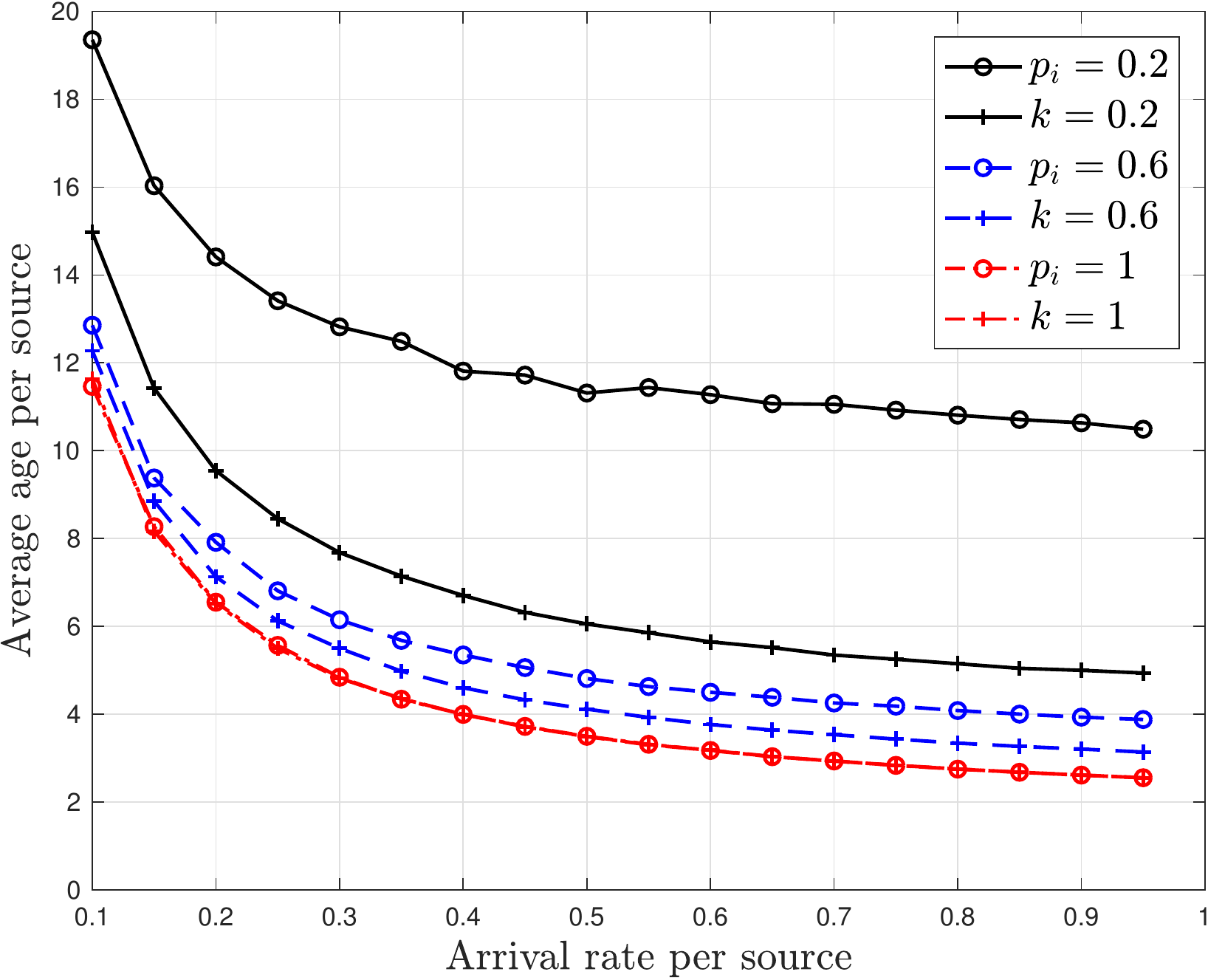}
	\caption{Average age per source vs. the arrival rate $\lambda_i$ for the round-robin scheduler with the replacement queue discipline, for $N=2$. The success probability of the $i$th node is $p_i$ and the network delay parameter is $k$.}
	\label{fig:age_vs_lambda_round_robin_replace_FadingandCloud_comparison}
\end{figure}

In Fig.~\ref{fig:age_vs_lambda_round_robin_replace_FadingandCloud_comparison} we compare the effect of the parameters $p_i$ and $k$ on the AoI objective.
Recall that the network delay process follows a geometric distribution with mean $1/k$.
The average AoI per source is shown as a function of the arrival rate per source node with a replacement queue, for the round-robin scheduler, $N=2$, and $\lambda_1=\lambda_2$.
For the different values of the success probability $p_i$ the AoI is measured at the AP.
For the different values of the parameter $k$ the AoI is measured at the destination $d$, assuming that the transmission to the AP is instantaneous end error-free.
We observe that the effect of the parameter $k$ differs from the effect of the parameter $p_i$.
This is due to the fact that a failure in transmission corresponds not only to a wasted time slot but also to a wasted turn for the source. 

In Fig.~\ref{fig:obsolete_packets_vs_k_round_robin_replace_withCloud} the number of obsolete packets is shown as a function of the network delay parameter $k$ for the round-robin scheduler, $N=2$, $\lambda_1=\lambda_2$, and success probability 1, at the destination $d$. 
As expected, increasing the arrival rate at the source nodes results in more packets that are rendered obsolete.
Hence, there is a tradeoff between the AoI performance and the number of wasted resources in terms of obsolete packets.

\begin{figure}[t!]\centering
	\centering
	\includegraphics[draft=false,scale=.5]{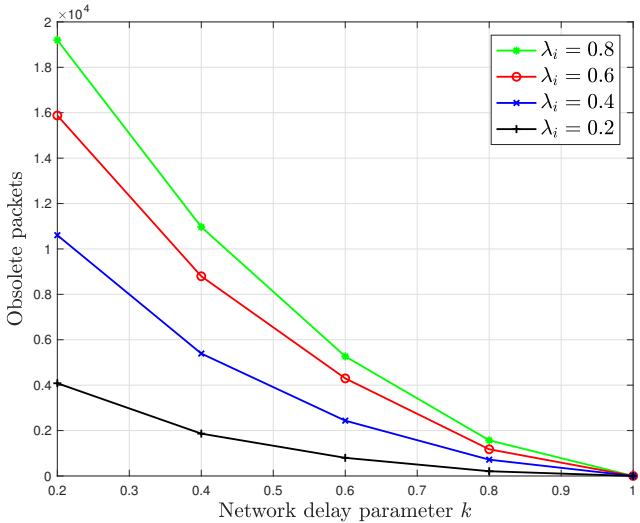}
	\caption{Obsolete packets vs. the network delay parameter $k$ for the round-robin scheduler with the replacement queue discipline, for $N=2$. The time horizon is 100000 time slots.}
	\label{fig:obsolete_packets_vs_k_round_robin_replace_withCloud}
\end{figure}

\section{Summary}
\label{sec:conclusions}

In this work, we have focused on the AoI performance of a network consisting of $N$ source nodes communicating with a common receiver.
Each source node has a buffer of infinite capacity to store incoming bursty traffic in the form of packets which should keep the receiver timely updated.
We have considered two different queue disciplines at the transmission queues, with and without packet management, and we have derived analytical expressions for the AoI for both cases.
We have investigated three different policies to access the common medium (i) round-robin scheduler (ii) work-conserving scheduler 
(iii) random access.
The work-conserving scheduler outperforms the round-robin scheduler.
For the case of the random access one should optimize the access probabilities in connection to the arrival rates per source and the number of source nodes in the system.
Moreover, we have considered the effect of the success probability of a packet erasure model and the effect of network path diversity, on the AoI performance.  The presented simulation results provide guidelines for the design of the system. 


\appendices
\section{Proof of Lemma \ref{lemma2}}\label{Appendix_A'}

Given the DTMC described in Fig.~\ref{fig:markov_chain_2} we define $r=\lambda(1-\mu)$ and $s=\mu(1-\lambda)$ and obtain the following balance equations:

\begin{align}
&\lambda \pi_0 = \mu (1-\lambda) \pi_1 \Leftrightarrow \pi_1 = \frac{\lambda}{\mu(1-\lambda)} \pi_0, \notag \\
& \pi_1 = \lambda \pi_0 + (1-r-s) \pi_1 + s \pi_2 \Leftrightarrow \pi_2 = \frac{\lambda^2 (1-\mu)}{\mu^2(1-\lambda)^2} \pi_0.  \notag
\label{eq:balance_eq_1_2}
\end{align}
Summarizing, for $n \in \{1,2\}$ we have that 
\begin{equation*}
\pi_n = \frac{\lambda^n (1-\mu)^{n-1}}{\mu^n (1-\lambda)^n} \pi_0. 
\label{eq:balance_eq}
\end{equation*}
Moreover, we know that 
\begin{equation*}
\pi_0 + \frac{\lambda}{\mu(1-\lambda)} \pi_0 + \frac{\lambda^2 (1-\mu)}{\mu^2(1-\lambda)^2} \pi_0 = 1.
\label{eq:sum_balance_eq}
\end{equation*}
Hence, the probability that the queue is empty is given by 
\begin{equation*}
\pi_0 = \frac{\lambda -\mu}{\lambda \left( \frac{\lambda(1-\mu)}{\mu(1-\lambda)}\right)^2-\mu}.
\label{eq:balance_eq_pi0}
\end{equation*}


\bibliography{references}
\bibliographystyle{IEEEtran}

\end{document}